\documentclass{amsart}
\usepackage{times}
\usepackage{amssymb,amsmath,amscd, amsthm,stmaryrd,verbatim, mathrsfs}

\usepackage{color}

\newtheorem{Thm}{Theorem}
\newtheorem{theorem}{Theorem}[section]

\newtheorem{Lem}[theorem]{Lemma}

\newtheorem{Prop}{Proposition}[section]

\newtheorem{rmk}{Remark}[section]

\numberwithin{equation}{section}

\newcommand{\bb}{\mathbb}

\newcommand{\mr}{\mathrm}
\newcommand{\frk}{\mathfrak}
\newcommand{\tr}{\mr{tr}\,}

\usepackage{hyperref}
\begin{document}

\title{The Poisson Realization of $\mathfrak{so}(2, 2k+2)$ on Magnetic Leaves and generalized MICZ-Kepler Problems}
\author{Guowu Meng}

\address{Department of Mathematics, Hong Kong Univ. of Sci. and
Tech., Clear Water Bay, Kowloon, Hong Kong}
%\curraddr{Institute for Advanced Study, Einstein Drive, Princeton, New Jersey 08540 USA}

\email{mameng@ust.hk}
\thanks{The author was supported by the Hong Hong Research Grants Council under RGC Project No. 603110 and the Hong Kong University of Science and Technology under DAG S09/10.SC02.}

%    General info
%\subjclass[2000]{Primary 22E46, 22E70; Secondary 81S99, 51P05}

\date{November 14, 2012}


\maketitle
\begin{abstract}
Let ${\mathbb R}^{2k+1}_*={\mathbb R}^{2k+1}\setminus\{\vec 0\}$ ($k\ge 1$) and $\pi$: ${\mathbb R}^{2k+1}_*\to \mathrm{S}^{2k}$ be the map sending $\vec r\in {\mathbb R}^{2k+1}_*$ to ${\vec r\over |\vec r|}\in \mathrm{S}^{2k}$.  Denote by $P\to {\mathbb R}^{2k+1}_*$ the pullback by $\pi$ of the canonical principal $\mathrm{SO}(2k)$-bundle $\mathrm{SO}(2k+1)\to \mathrm{S}^{2k} $. Let $E_\sharp\to {\mathbb R}^{2k+1}_*$ be the associated co-adjoint bundle and $E^\sharp\to T^*{\mathbb R}^{2k+1}_*$ be the pullback bundle under projection map $T^*{\mathbb R}^{2k+1}_*\to {\mathbb R}^{2k+1}_*$. The canonical connection on $\mathrm{SO}(2k+1)\to \mathrm{S}^{2k} $ turns $E^\sharp$ into a Poisson manifold. 

The main result here is that the real Lie algebra $\mathfrak{so}(2, 2k+2)$ can be realized as a Lie subalgebra of the Poisson algebra $(C^\infty(\mathcal O^\sharp), \{, \})$, where $\mathcal O^\sharp$ is a symplectic leave of $E^\sharp$ of special kind. Consequently, in view of the earlier result of the author, an extension of the classical MICZ Kepler problems to dimension $2k+1$ is obtained. The hamiltonian, the angular momentum, the Lenz vector and the equation of motion for this extension are all explicitly worked out.

 \end{abstract}

\tableofcontents

\section {Introduction}
The Kepler problem is the physics problem about two bodies which attract each other by a force proportional to the inverse square of
the distance. Historically this problem played a pivotal role in the development of both classical mechanics and quantum mechanics.

The Kepler problem has long been known to exist in all higher dimensions. A surprising discovery \cite{MC70} in the 1960s is that the magnetized versions of the Kepler problem, under the name of MICZ Kepler problem, also exist. For a while, these magnetized versions were thought \cite{Iwai90} to exist only in dimension $3$, $5$ and possibly $9$, corresponding to the division algebra $\mathbb C$, $\mathbb H$ and $\mathbb Q$ respectively. However, it was demonstrated in Ref. \cite{meng07} that these magnetized versions exist quantum mechanically in all higher dimensions, and that leads us to believe that these magnetized versions also exist classically in all higher dimensions.  

In Ref. \cite{meng09} the intimate relation between the conformal algebras of euclidean Jordan algebras and the generalized Kepler problems (i.e., integrable models which share the characteristic features of the Kepler problem) was discovered. In particular, this relation says that the ``nice" Poisson realization of the conformal algebra on symplectic spaces and the classical generalized Kepler problems correspond to each other. 

As an elaboration of the last sentence, let us take the Kepler problem as an example. Here, the euclidean Jordan algebra is $V=\Gamma(3):=\mathbb R\oplus \mathbb R^3$ with this Jordan multiplication: 
$$(\alpha, \mathbf u)(\beta, \mathbf v)=(\alpha\beta+\mathbf u\cdot \mathbf v, \alpha\mathbf v+\beta\mathbf u),$$
and the conformal algebra is $\frk{so}(2,4)$. Let $C_1$ be the future light cone for the Minkowski space. Using the standard euclidean structure on $V$, $T^*C_1$ can be viewed as a symplectic submanifold of $T^*V$. It is well-known that $\frk{so}(2,4)$ has a Poisson realization on the total cotangent space $T^*V$, hence a Poisson realization on $T^*C_1$. This later Poisson realization on $T^*C_1$, or equivalently on $T^*(\mathbb R^3\setminus\{0\})$, is the ``nice" Poisson realization of $\frk{so}(2,4)$ that corresponds to the classical Kepler problem. Though not emphasized in this article, the reformulation of the Kepler problem as a dynamic problem on $C_1$ is much more natural, cf. Ref. \cite{meng11}.  

\vskip 10pt
In section \ref{magnetic cone}, we introduce the notion of {\bf magnetic cone} for the Lie group $\mathrm{SO}(2k)$, a real algebraic set inside $\frk{so}^*(2k)$ which is the union of certain special co-adjoint orbits of $\mathrm{SO}(2k)$.  In section \ref{phase space}, we review the notion of Wong's phase space and its symplectic leaves (i.e., Sternberg's phase spaces) \cite{Sternberg77, Weinstein78, Montgomery84}, and then introduce the notion of {\bf magnetic leaves}, which are certain special kind of Sternberg's phase spaces. In section \ref{monopoles}, we review the notion of generalized Dirac monopoles \cite{Cotaescu05, meng04}, a concept that is crucially used in the introduction of the magnetic leaves. In section \ref{PoissonRelization}, we describe and prove our main result: the Poisson realization of the real Lie algebra $\frk{so}(2, 2k+2)$ on the magnetic leaves. In the last section, based on Ref. \cite{meng11a} and the main result here, we derive an extension of the classical MICZ Kepler problems to dimension $2k+1$. The hamiltonian, the angular momentum, the Lenz vector and the equation of motion for this extension are all explicitly worked out.

\section{Magnetic cone}\label{magnetic cone}
Let $G$ be a compact semi-simple Lie group, $\frak g$ be its Lie algebra and $\frak h$ be (one of) its Cartan subalgebra. The Lie bracket shall be denoted by $[, \,]$, and the natural pairing ${\frk g}^*\otimes {\frk g}\to \mathbb R$ shall be written as $<,\, >$. By convention, the elements of $i\frak g$ (rather than $\frk g$) shall be viewed as linear functions on $\frak g^*$: for $X\in i\frk g$ and $\xi\in \frk g^*$ we have 
\begin{eqnarray}
X(\xi):=<\xi, iX>.
\end{eqnarray}
It is well known that the Lie algebra structure on $\frk g$ defines a natural Poisson structure on $\frk g^*$: Let $X, Y\in i\frak g$, then the Poission bracket of the linear function $X$ with the linear function $Y$ is the linear function $i[iX, iY]$, i.e., 
\begin{eqnarray}
\{X, Y\}:=i[iX, iY].
\end{eqnarray}
The adjoint action of $G$ on $\frk g$ gives rise to the co-adjoint action of $G$ on $\frk g^*$. This co-adjoint action partitions $\frak g^*$ into a disjoint union of the co-adjoint orbits. Except for the trivial orbit $\{0\}$, the co-adjoint orbits are symplectic manifolds and are in fact the symplectic leaves of the Poisson manifold $\frk g^*$. 

The irreducible representations of $G$ are all finite dimensional and unitarizable, and are parametrized by their highest weights, which are precisely the dominant integral weights for $G$. A nice way to look at these representations is via the orbit theory of Kirillov \cite{Kirillov04}. Since $G$ is a compact semi-simple Lie group, its co-adjoint orbits are closed and each of them intersects the fundmental Weyl chamber $\frk{C}\subset \frk{h}^*$ in a single point. (Using the Killing metric on $\frk g$, $\frk{h}^*$ can be naturally imbedded into $\frk{g}^*$, so $\frk{C}\subset \frk{g}^*$.) An orbit $\mathcal O$ is called integral if the unique point in $\mathcal{O}\cap \frk{C}$ belongs to the weight lattice of $G$. The highest weight theory can be restated in the form of a bijection between the set of integral co-adjoint orbits and the set of equivalence classes of irreducible unitary representations of $G$: the highest weight representation $L(\lambda)$ with highest weight $\lambda\in\frk{C}$ corresponds to the integral co-adjoint orbit $G\cdot \lambda$.

\vskip 10pt
In the remainder of this section, we take $G=\mathrm{SO}(2k)$, $\frk{g}=\frk{so}(2k)$, $\frk{h}=\oplus_{i=1}^k \frk{so}(2)$ so that an element of $\frk h$ in the defining representation of $\frk{so}(2k)$ is a block-diagonal matrix whose diagonal blocks are $2\times 2$ real skew-symmetric matrices. Hereafter, we shall use $\gamma_{ab}$ ($1\le a, b\le 2k$) to denote the element of $i\frk{g}$ such that in the defining representation of $\frk{g}$, $i\gamma_{ab}$ is represented by the skew-symmetric real symmetric matrix whose $(a, b)$-entry is $-1$, $(b, a)$ entry is $1$, and all other entries are $0$. For the invariant metric on $\frk{g}$, we take the following one: Let $\rho$ be the defining representation of $\frk{g}$, for $\xi$, $\eta$ in $\frk{g}$, we have
\begin{eqnarray}\label{Killing}
(\xi, \eta)=-{1\over 2}\tr \left(\rho(\xi)\, \rho(\eta)^T\right). 
\end{eqnarray} Then $i\gamma_{ab}$ ($1\le a< b\le 2k$) form an orthonormal basis for $\frk g$. One can verify that the following Poisson bracket relations hold:
\begin{eqnarray}\label{PoissonRelGamma}
\{\gamma_{ab},\gamma_{cd}\}=-\delta_{ac}\gamma_{bd}+\delta_{bc}\gamma_{ad}-\delta_{bd}\gamma_{ac}+\delta_{ad}\gamma_{bc}.
\end{eqnarray}

Let $\sigma_+=\sum_{a=1}^k i\gamma_{2a-1, 2a}\, \in \frk{h}$. Via the invariant metric (\ref{Killing}), $\sigma_+$ is turned into an element of $\frk g^*$, which shall be denoted by $\sigma_+^*$. Let $G'=G\times \mathbb R_+$. The co-adjoint action of $G$ on $\frk g^*$ and the scaling action on $\frak g^*$ commute, so they gives rise to an action of $G'=G\times \mathbb R_+$ on $\frk g^*$. The orbit $G'\cdot \sigma_+^*$ shall be referred to as the {\bf positive magnetic cone} $\mathcal M_+$ for $G$.

Let $\sigma_-=i\gamma_{12}+\cdots+i\gamma_{2k-3, 2k-2}-i\gamma_{2k-1, 2k}$. Via the invariant metric (\ref{Killing}), $\sigma_-$ is turned into an element of $\frk g^*$, which shall be denoted by $\sigma_-^*$. The orbit $G'\cdot \sigma_-^*$ shall be referred to as the {\bf negative magnetic cone} $\mathcal M_-$ for $G$.

One can check that $\mathcal M_+$ and $\mathcal M_-$ are disjoint from each other, and $0\in\frk g^*$ is the unique extra limit point of $\mathcal M_+$ and also $\mathcal M_-$. We denote by $\mathcal M$ the union of these three $G'$ orbits: $\mathcal M_+$, $\{0\}$, and $\mathcal M_-$.  Then $\mathcal M$ is topologically closed and shall be referred to as the {\bf magnetic cone} for $G$.  

We shall show that the magnetic cone for $G$ is a real algebraic set inside $\frk{g}^*$. To do that, we let $Q$ be the homogenous quadratic polynomial function ${1\over 2k}\sum_{1\le a, b \le 2k}(\gamma_{ab})^2$ on $\frk g^*$. Then $Q$ is invariant under the action of $G$: Let $g\in G$, then $g\cdot \gamma_{ab}=g_{aa'}g_{bb'}\gamma_{a'b'}$ where $[g_{aa'}]\in \mathrm{SO}(2n)$, so
\begin{eqnarray}
g\cdot Q &=& {1\over 2k}\sum_{1\le a, b \le 2k}(g\cdot \gamma_{ab})^2\cr
&=& {1\over 2k}\sum_{1\le a, b \le 2k} g_{aa'}g_{bb'}\gamma_{a'b'} g_{aa''}g_{bb''}\gamma_{a''b''} \cr
&=& {1\over 2k} \delta_{a'a''}\delta_{b'b''}\gamma_{a'b'} \gamma_{a''b''} \cr
&=& {1\over 2k} \sum_{1\le a', b' \le 2k}\gamma_{a'b'} \gamma_{a'b'} = Q\nonumber.
\end{eqnarray} A similar computation shows that if $\sum_{a=1}^{2k}\gamma_{ab}\gamma_{ac}=\delta_{bc} Q$ hold for $1\le b, c\le 2k$, then  $\sum_{a=1}^{2k}(g\cdot \gamma_{ab})(g\cdot \gamma_{ac})=\delta_{bc}Q$ hold for $1\le b, c\le 2k$ and any $g\in G$, in fact, for any $g\in G'$. 

\begin{Prop}\label{propQE}
The magnetic cone $\mathcal M$ for $\mathrm{SO}(2k)$ is the real algebraic set defined by the homogenous quadratic equations
$$
\sum_{a=1}^{2k}\gamma_{ab}\gamma_{ac}=\delta_{bc}Q \quad 1\le b, c\le 2k.
$$
\end{Prop}
\begin{proof}  1) We need to check that the identities hold on $\mathcal M$. Let us do the checking on $\mathcal M_+$. Due to the remark we made right before the statement of this proposition, it suffices to
verify the identities at point $\sigma_+^*$, i.e., the identities
$$
\sum_{a=1}^{2k}(i\gamma_{ab}, \sigma_+)( i\gamma_{ac}, \sigma_+)={\delta_{bc}\over k} \sum_{1\le a < b \le 2k}(i\gamma_{ab}, \sigma_+)^2. 
$$ These identities are clearly true because $\sigma_+=i\gamma_{12}+i\gamma_{34}+\cdots+i\gamma_{2k-1, 2k}$. Similarly, the identities hold on $\mathcal M_-$. Therefore, the identities hold on the magnetic cone. 

2) We need to check that the identities hold only on $\mathcal M$. It suffices to check that if the identities \begin{eqnarray}\label{identities}
\sum_{a=1}^{2k}(i\gamma_{ab}, \sigma)(i \gamma_{ac}, \sigma)={\delta_{bc}\over k} \sum_{1\le a < b \le 2k}(i\gamma_{ab} \sigma)^2,
\end{eqnarray} hold for $\sigma = x_1i\gamma_{12}+x_2i\gamma_{34}+\cdots+x_ki\gamma_{2k-1, 2k}$, then $|x_1| = |x_2| =\cdots = |x_k|$, so $\sigma^*\in \mathcal M$. This is indeed the case because, by letting $b=c=2j$ in identities (\ref{identities}), we have
$$
x_j^2={1\over k}(x_1^2+\cdots+x_k^2)
$$ which is the same for all $j$ between $1$ and $k$.
\end{proof}
We conclude this section by introducing a magnetic charge function $m$: $\mathcal M\to \mathbb R$. By definition, $|m(\xi)|=|\xi|$ and the sign of $m(\xi)$ is $\pm$ if $\xi\in \mathcal M_\pm$. It is clear that $m$ is invariant under the co-adjoint action of $G$ and an orbit inside the magnetic cone is uniquely determined by the constant value that $m$ takes. For this reason, if an orbit on which $m$ takes the constant value $\mu$, then this orbit is denoted by $\mathcal O_\mu$ and is referred to as the {\bf magnetic orbit with magnetic charge $\mu$}. Note that $\mathcal O_\mu$ is a point if $\mu=0$ and is diffeomorphic to the compact hermitian symmetric space $\mr{SO}(2k)/\mr{U}(k)$ if $\mu\neq 0$. It is a well-known fact that this compact hermitian symmetric space is the space of complex structures on $\mathbb R^{2k}$ compatible with the standard inner product on $\mathbb R^{2k}$. 

The special symplectic leaves $\mathcal O_\mu$ are the ones and the only ones on which the homogenous quadratic equations in proposition \ref{propQE} are satisfied, or equivalently they are the ones and the only ones such that part 2) of Lemma \ref{lemma} or Theorem \ref{main} holds.

\section{Sternberg's phase spaces and magnetic leaves}\label{phase space}
This section is mainly a review of Sternberg's phase spaces, cf. Refs \cite{Sternberg77, Weinstein78, Montgomery84}.  Let $X$ be a smooth $n$-manifold, $G$ a compact connected Lie group with Lie algebra $\frk g$, and $P\to X$ a principal $G$-bundle over $X$ with a fixed connection $\nabla$. Denote by $P^\sharp\to T^*X$ the pullback bundle of this principal $G$-bundle by the bundle projection map $T^*X\to X$, and 
\begin{eqnarray}
E^\sharp\to T^*X
\end{eqnarray} be the co-adjoint bundle, i.e., the associated real vector bundle with $\frk g^*$ as its fiber. 
\begin{Prop}
With the set-up given above, the connection $\nabla$ turns $E^\sharp$ into a Poisson manifold. 
\end{Prop}
As shown in Ref. \cite{Montgomery84}, the Hamilton's equation on $E^\sharp$ for a natural Hamiltonian function is Wong's equation in Ref. \cite{Wong70}, so $E^\sharp$ is referred to as the Wong's phase space. However, the more relevant phase spaces to us are Sternberg's phase spaces $\mathcal O^\sharp:= P^\sharp\times_G\mathcal O\subset E^\sharp$, which are known to be the symplectic leaves of the Wong's phase space $E^\sharp$, cf. Ref. \cite{Montgomery84}.

We shall produce an indirect proof of this proposition by describing the Poisson bracket relations among the local coordinate functions. We use $(U, \phi)$ to denote a local coordinate chart for $X$ and write $\phi=(x^1, \ldots, x^n)$, then the cotangent frame $dx^1$, \ldots, $dx^n$ over $U$ gives a local trivialization of the tangent bundle $T^*X\to X$ over $U$. This local trivialization shall always be assumed hereafter. By choosing $U$ small enough if necessary, we can assume that the principal $G$-bundle $P\to X$ is trivial over $U$. We shall fix a trivialization of the bundle $P\to X$ over $U$, then the connection $\nabla$ can be represented by a $i\frk{g}$-valued differential one-form $A$ on $U$, and the curvature can be represented by a $i\frk{g}$-valued differential two-form $F$ on $U$. We shall also fix a basis $\{T_a\}_{i=1}^D$ for $i\frk{g}$ where $D=\dim g$. Since each $T_a$ is a linear function on $\frk{g}^*$ by our convention, it is now clear that we have a local coordinate map from $E^\sharp$ to $\phi(U)\times {\mathbb R^n}\times {\mathbb R}^{D}={\mathbb R}^{2n+D}$. The domain of this coordinate map shall be referred to as a \emph{good coordinate patch} for $E^\sharp$. On such a good coordinate patch we have the following local coordinate functions: $$(x^1, \ldots, x^n, \pi_1, \ldots, \pi_n, T_1, \ldots, T_D).$$
In terms of this local coordinate functions on $E^\sharp$, the basic Poisson bracket relations are
\begin{eqnarray}
\left\{
\begin{array}{lll}
\{x^j, x^k\}=0, & \{x^j,\pi_k\}=\delta_k^j, & \{\pi_j, \pi_k\}=-F_{jk}, \\
\\
\{T_\alpha, T_\beta\}=-C_{\alpha\beta}^\gamma T_\gamma, & \{T_\alpha, x_k\}=0,&
 \{T_\alpha, \pi_k\}= - C_{\alpha\beta}^\gamma A_k^\beta T_\gamma,
 \end{array}\right.
\end{eqnarray} where $C_{\alpha\beta}^\gamma$ s are the structure constants of $\frk g$ with respect to basis $iT_\alpha$. Note that we have the vector potential $A=A_k\, dx^k$ with $A_k=A_k^\alpha T_\alpha$, the gauge field $F:=dA+iA\wedge A ={1\over 2}F_{jk}\, dx^j\wedge dx^k$ with $F_{jk}=F_{jk}^\alpha T_\alpha$, so, if we use $p=p_k\, dx^k$ to denote the canonical momentum, we have $\pi_k=p_k+A_k$. Note also that $A_k^\alpha$ and $F_{jk}^\alpha$ are all real-valued functions of $x^1$, ... , $x^n$. Finally, we remark that $iF=\Omega$ and $iA=\omega$ where $\Omega$ is the curvature form and $\omega$ is the connection form.

We conclude this section by introducing a few more spaces: $E_\sharp:=P\times_G \frk{g}^*$, $\mathcal O_\sharp:=P\times_G \mathcal O$. In the case $G=\mathrm{SO}(2k)$,  $X$ is $\mathbb R^{2k+1}_*:=\mathbb R^{2k+1}\setminus\{0\}$, the co-adjoint orbit $\mathcal O$ is a magnetic orbit, and the connection on $P\to X$ is a generalized Dirac monopole, a key concept which shall be reviewed in the next section, we let ${\mathcal M}^\sharp:=P^\sharp \times_G {\mathcal M}$ and ${\mathcal M}_\sharp=P \times_G {\mathcal M}$. (recall that $\mathcal M$ denotes the magnetic cone for $\mathrm{SO}(2k)$.)  It will be clear in the last section that there is a bijection between the symplectic leaves of ${\mathcal M}^\sharp$ and the classical Kepler problems (with magnetic charges) in dimension $2k+1$: the phase space of the classical Kepler problem with magnetic charge $\mu$ is the symplectic leave $\mathcal {O_\mu}^\sharp$, where $O_\mu$ is the magnetic orbit with magnetic charge $\mu$. A symplectic leave of $E^\sharp$ of the form $\mathcal {O_\mu}^\sharp$ shall be referred to as the {\bf magnetic leave} of magnetic charge $\mu$. Note that the dimension of this magnetic leave is $4k+2$ if $\mu=0$ and is $k^2+3k+2$ if $\mu\neq 0$. 

 \section {Generalized Dirac monopoles}\label{monopoles}
The Sternberg's phase spaces discussed in the above shall be specialized to our needs. In this specialization, $G=\mathrm{SO}(n-1)$, $X$ is $\mathbb R^n_*:=\mathbb R^n\setminus\{0\}$, the co-adjoint orbit $\mathcal O$ is a magnetic orbit, and the connection on $P\to X$ is a generalized Dirac monopole \cite{Cotaescu05, meng04}.
 
The discussion of generalized Dirac monopoles starts with the principal $G$ bundle over $\mathrm{S}^{n-1}$:
$$\begin{array}{c}
 \mathrm{SO}(n)\cr
 \Big\downarrow\cr
 \mathrm{S}^{n-1}. 
\end{array}
 $$This bundle comes with a natural connection
$$\omega (g):=\mathrm{Pr}_{\frk{so}(n-1)}\left(g^{-1}dg\right),$$
where $g^{-1}dg$ is the Maurer-Cartan form for $ \mathrm{SO}(n)$, so it is an $\frk{so}(n)$-valued differential one form on $\mathrm{SO}(n)$, and $\mathrm{Pr}_{\frk{so}(n-1)}$ denotes the orthogonal projection of $\frk{so}(n)$ onto $\frk{g}:=\frk{so}(n-1)$.

Under the map
\begin{eqnarray}
\pi:  \mathbb R^n_* & \to & \mathrm{S}^{n-1}\cr
 \vec  r &\mapsto &{\vec r \over |\vec r|},
\end{eqnarray}
the above bundle and connection are pulled back to a principal $G$-bundle 
\begin{eqnarray}
\begin{array}{c}
 P\cr
 \Big\downarrow\cr
  \mathbb R^n_* 
\end{array}
\end{eqnarray} with a connection. This connection, originally introduced in Ref. \cite{meng04} and also independently in Ref. \cite{Cotaescu05}, extends the Dirac monopole to high dimensions, is referred to as the generalized Dirac monopole in dimension $n$.
 
We shall write $\vec r =(x^1, \ldots, x^n)$ for a point in $\mathbb R^n_*$ and $r$ for the length of $\vec r$. Sometime $x^i$ is also written as $x_i$. The small Lartin letters $j$, $k$, etc. be indices that run from $1$ to $n$, and the small Latin letters $a$, $b$, etc. be indices that run from $1$ to $n-1$. To do local computations, we need to choose a trivialization on $U$ which is $\mathbb R^n$ with the negative $n$-th axis removed and then write down the gauge potential explicitly. We have done that before in Eq. (10) of Ref.
\cite{meng04}.  Note that, the gauge potential $A=A_k\, dx^k$
from Eq. (10) of Ref. \cite{meng04} can be written as
\begin{eqnarray}\label{mnple}
A_n =0,\hskip 20 pt A_b=-{1\over r(r+x_n)}x^a\gamma_{ab}
\end{eqnarray}
where $\gamma_{ab}={i\over 4}[\gamma_a,\gamma_b]\in i\frk g$ with $\gamma_a$
being the ``gamma matrix" for physicists, $x^a\gamma_{ab}$ means $\sum_{a=1}^{n-1}x^a\gamma_{ab}$, something we shall assume whenever there is a repeated index.  Note that $\gamma_a=i e_a$
with $e_a$ being the element in the Clifford algebra that corresponds to the $a$-th standard coordinate vector of $\bb
R^{n-1}$.

It is straightforward to calculate the gauge field strength
$F_{jk}=
\partial_jA_k-\partial_k A_j + i [A_j, A_k]$ and get
\begin{eqnarray}
F_{nb} &=& {1\over r^3}x^a\gamma_{ab},\\
F_{ab}& = & -{1\over r^2}\gamma_{ab} + {x_a
x^c\gamma_{cb}-x_b x^c\gamma_{ca}\over
r^3(r+x_n)}\cr
&= & -{1\over r^2}(\gamma_{ab} +x_aA_b - x_bA_a).
\end{eqnarray}

Hereafter we assume that $n=2k+1$. The following lemma, whose quantum version first appeared in Ref. \cite{meng07}, is crucially used in this article.
\begin{Lem}\label{lemma}
Let $Q={1\over 2k}\sum_{a, b}(\gamma_{ab})^2$ and $\nabla_k=\partial_k+iA_k$. For the gauge potential $A$ defined in Eq. (\ref{mnple}),  the following statements are true.

1) As functions on the good coordinate patch, 
\begin{eqnarray}\label{Id1}
x^k A_k=0,\quad x^j F_{jk}=0,\quad r^4\sum_{i,j}(F_{ij})^2=2kQ
\end{eqnarray}
\begin{eqnarray}\label{Id2}
\nabla_l F_{jk}={1\over r^2}\left( -x^j
F_{lk}-x^k F_{jl}-2x^l F_{jk} \right), \quad \sum_j \nabla_j F_{jk}=0
\end{eqnarray}
\begin{eqnarray}\label{Id3}
r^4\{F_{jk}, F_{lm}\}=(r^2\delta_{jl}-x_j x_l)F_{km}+\mbox{the remaining three terms}.
\end{eqnarray}

2) The identities
\begin{eqnarray}\label{Id4}
r^4\sum_iF_{ij} F_{ij'}=Q\left(\delta_{jj'}-{x_j x_{j'}\over
r^2}\right)
\end{eqnarray} hold at the point $[(p^\sharp, \xi)]\in E^\sharp$ if and only if the following set of homogeneous quadratic equations $\sum_a \gamma_{ab}\gamma_{ab'}=\delta_{bb'}Q$ hold at $\xi$, i.e., if and only if $\xi$ is inside the magnetic cone $\mathcal M$. 

\end{Lem}
\begin{proof} Since $\gamma_{ab}=-\gamma_{ba}$, we have
$$
\sum_j x_jA_j =\sum_b x_bA_b=-{1\over r(r+x_n)}\sum_{a, b}x_ax_b\gamma_{ab}=0.
$$
The verification of the remaining identities is a direct and lengthy calculation. However, we note that the connection
is rotational invariant, i.e., under the rotations of $\mathbb R^{2k+1}$, $A$ (hence $F$) is invariant modulo gauge transformations. Combining with the transformation property of the remaining identities under rotations and gauge transformations, we just need to check the remaining identities at point
$\vec r_0=(0,\ldots, 0, r)$, a much easier task. At this point $\vec r_0$, since
\begin{eqnarray} 
A_i =0, \quad F_{na}=0, \quad F_{ab}=-{1\over r^2}\gamma_{ab},
\end{eqnarray} one can quickly finish the checking. 
\end{proof}

\section{The Poisson realization of $\frk{so}(2, 2k+2)$ on magnetic leaves}\label{PoissonRelization}
The goal in this section is to present and prove our main result:  the Poisson realization of $\frk{so}(2, 2k+2)$ on the magnetic leave ${\mathcal O_\mu}^\sharp$. We take the manifold $X$ to be $\mathbb R^{2k+1}_*$, the coordinate patch $U$ for $X$ to be $\mathbb R^{2k+1}$ with the negative $(2k+1)$-st axis removed, and the coordinate map $\phi$: $U\to \mathbb R^{2k+1}$ to be the inclusion, and the bundle $P\to X$ and its trivialization over $U$ to be the one given in the last section.  Note that this good coordinate patch for $E^\sharp$ is dense in $E^\sharp$. We let $x_j=x^j$, ${\vec r}=(x_1, \ldots, x_{2k+1})$ and $r=|\vec r|$, ${\vec \pi}=(\pi_1, \ldots, \pi_{2k+1})$ and $\pi=|\vec \pi|$. Recall from section \ref{phase space} that the basic Poisson bracket relations are
\begin{eqnarray} \label{PoissonRel}
\left\{
\begin{array}{rcl}
\{x_j, x_k\}=0, & \{x_j, \pi_k\}=\delta_{jk}, & \{\pi_j, \pi_k\}=-F_{jk}, \\
\\
\{T_\alpha, T_\beta\}=-C_{\alpha\beta}^\gamma T_\gamma, & \{T_\alpha, x_k\}=0, &
 \{T_\alpha, \pi_k\}= - C_{\alpha\beta}^\gamma A_k^\beta T_\gamma.
\end{array}\right.
\end{eqnarray}
Let us first introduce
\begin{eqnarray}\left\{\begin{array}{rcl}
J_{i, 0} &: =& r\pi_i,\cr 
J_{2k+2, 0} &: = & {1\over 2} (r  \pi^2+{Q\over r}) -{1\over 2}r,\cr
J_{-1, 0} &:= & {1\over 2} (r  \pi^2+{Q\over r})+ {1\over 2}r.
\end{array}\right.
\end{eqnarray} where $Q={1\over 2k}\sum_{a, b}(\gamma_{ab})^2$ as in Lemma \ref{lemma}.
Then we let
\begin{eqnarray}\left\{\begin{array}{rcl}
J_{i,j}  &:= & -\{J_{i, 0}, J_{j, 0} \},\cr 
J_{i, 2k+2} &:= &-\{J_{i,0}, J_{2k+2, 0}\},\cr 
J_{i, -1}&:= &-\{J_{i,0}, J_{-1, 0}\},\cr
J_{2k+2, -1} &:= & -\{J_{2k+2,0}, J_{-1, 0}\} .
\end{array}\right.
\end{eqnarray}
By a small computation based on the basic Poisson bracket relations (\ref{PoissonRel}), one arrives at the following explicit expression:
\begin{eqnarray}\left\{\begin{array}{rcl}
J_{i,j}  &= & x_i\pi_j - x_j\pi_i+r^2F_{ij},\cr 
J_{i, 2k+2} &= &{1\over 2}x_i\pi^2 - \pi_i (\vec r\cdot  \vec \pi)+r^2F_{ij}\pi_j-{Q\over 2r^2}x_i-{1\over 2}x_i,\cr 
J_{i, -1}&= &{1\over 2}x_i\pi^2 - \pi_i (\vec r\cdot  \vec \pi)+r^2F_{ij}\pi_j-{Q\over 2r^2}x_i+{1\over 2}x_i,\cr 
J_{2k+2, -1} &= & \vec r\cdot\vec  \pi.
\end{array}\right.
\end{eqnarray} Here we assume the repeated indices are dummy, i.e., are summed up.

Let the capital Latin letters such as $A$, $B$, etc. be indices that run from $-1$ to $2k+2$. Note that, as a real function on $E^\sharp$, $F_{ij}$ factorizes through $E_\sharp$, moreover, for any point $[(p, \xi)]\in E_\sharp=P\times _G \frk{g}^*$, we have \begin{eqnarray}\label{F-mean}
F_{ij}([(p, \xi)]) =<\xi, F_{ij}(\pi(p))>.
\end{eqnarray} 
Here, $\pi$: $P\to X$ is the bundle projection.
Therefore $J_{AB}$ s are independent of the local trivialization we have fixed for the bundle $P\to X$, so they are defined on the whole $E^\sharp$, not just on a dense subset of $E^\sharp$. Recall that $\mathcal M^\sharp=P^\sharp\times_G \mathcal M$ is the union of all magnetic leaves of $E^\sharp$.

\begin{Thm}\label{main} 
Viewing $J_{AB}$ s as functions on $\mathcal M^\sharp$, we have the following two statements. 

1) $J_{AB}$ s satisfy the following Poisson bracket relations:
\begin{eqnarray}\label{cmtr} \{J_{AB},
J_{A'B'}\}=-\eta_{AA'}J_{BB'}-\eta_{BB'}J_{AA'}+\eta_{AB'}J_{BA'}+\eta_{BA'}J_{AB'}
\end{eqnarray}
where the indefinite metric tensor $\eta$ is ${\mr
{diag}}\{1,1, -1,\ldots, -1\}$ relative to the following order: $-1$, $0$,
$1$, \ldots, $2k+2$ for the indices.

2)  $J_{AB}$ s satisfy the following quadratic relations
\begin{eqnarray}\label{QuadraticRel} \eta^{AA'} J_{AB} J_{A'C}=\eta_{BC}Q.
\end{eqnarray} 
\end{Thm}
\begin{proof} One just needs to prove the two statements in a dense subset of $\mathcal M^\sharp$, such as the intersection of $\mathcal M^\sharp$ with the dense good coordinate patch for $E^\sharp$ mentioned in the beginning paragraph of this section. 
 
The quadratic relations in part 2) of this theorem can be reduced to a single one (referred to as the primary quadratic relation in Ref. \cite{meng09}):
\begin{eqnarray}\label{primary} \sum_{i=1}^{2k+1} J_{i,0}^2 +J_{2k+2, 0}^2-J_{-1, 0}^2=-Q
\end{eqnarray} because of the following observation initially discovered in Ref. \cite{meng09}:  thanks to the Poisson bracket relations in part 1) of Theorem \ref{main}, all the other quadratic relations in part 2) can be obtained by taking the Poisson bracket of the primary quadratic relation with some suitable $J_{AB}$. The checking of the primary quadratic relation (\ref{primary}) is very easy:
\begin{eqnarray}
 \sum_{i=1}^{2k+1} J_{i,0}^2 +J_{2k+2, 0}^2-J_{-1, 0}^2 =r\pi_i r\pi_i-XY= r^2\pi^2-(r^2\pi^2+Q)=-Q .\nonumber
\end{eqnarray}  
Therefore part 2) is proved modulo part 1). However, the proof of part 1) is quite involved, and the next subsection is wholly devoted to it.
\end{proof}

\subsection{Proof of part 1) of theorem \ref{main}}
The following lemma is quite useful in the proof of part 1) of Theorem \ref{main}.
\begin{Lem}\label{lemma2}
\begin{eqnarray}\left\{\begin{array}{rcl}
\{J_{ ij}, x_k\}& = &
-x_ i\delta_{jk}+x_j\delta_{ ik},\cr
\{J_{ ij}, \pi_k\}& = &
-\pi_ i\delta_{jk}+\pi_j\delta_{ ik},\cr
\{J_{ ij}, F_{ i'j'}\} &= &
\delta_{ i i'}F_{jj'}+\delta_{jj'}F_{ i i'}
-\delta_{ ij'}F_{j i'}-\delta_{j i'}F_{ ij'}.\end{array}\right.
\end{eqnarray}
\end{Lem}
\begin{proof} 
\begin{eqnarray}
\{J_{ ij}, x_k\}& = &
\{x_ i\pi_j-x_j\pi_ i+r^2F_{ ij}, x_k\}=\{x_ i\pi_j-x_j\pi_ i,x_k\}\cr
&= &
-(x_ i\delta_{jk}-x_j\delta_{ ik}).\nonumber
\end{eqnarray}
\begin{eqnarray}
\{J_{ ij}, \pi_k\}& = &
\{x_ i\pi_j-x_j\pi_ i+r^2F_{ ij}, \pi_k\}\cr
&= &
\pi_j\delta_{ ik} -\pi_ i\delta_{jk}-x_ i
F_{jk}+x_j F_{ ik}+\nabla_k (r^2F_{ ij})\cr &=&
-\pi_ i\delta_{jk}+\pi_j\delta_{ ik}.\quad \mbox{By identity (\ref{Id2}) }\nonumber
\end{eqnarray}
\begin{eqnarray}
\{J_{ ij},F_{ i'j'}\} & = &
\{x_ i\pi_j-x_j\pi_ i+r^2F_{ ij},
F_{ i'j'}\}\cr &=& x_ i\{\pi_j,
F_{ i'j'}\}-x_j
\{\pi_ i,F_{ i'j'}\}+r^2\{F_{ ij},
F_{ i'j'}\} \cr &=& {x_ i\over r^2}(2x_j
F_{ i'j'}+x_{ i'} F_{jj'}+x_{j'}
F_{ i'j})-{x_j\over r^2} (2x_ i
F_{ i'j'}+x_{ i'} F_{ ij'}+x_{j'}
F_{ i' i})\cr & &+r^2\{F_{ ij}, F_{ i'j'}\}\quad\mbox{by identity (\ref{Id2})}
\cr &=& r^2\{F_{ ij}, F_{ i'j'}\}-{1\over
r^2}\left(-x_ i x_{ i'} F_{jj'}-x_ i x_{j'}
F_{ i'j}+x_j x_{ i'} F_{ ij'}+x_j
x_{j'} F_{ i' i}\right)\cr & = &
\delta_{ i i'}F_{jj'}+\delta_{jj'}F_{ i i'}
-\delta_{ ij'}F_{j i'}-\delta_{j i'}F_{ ij'}.\quad\mbox{by identity (\ref{Id3})}\nonumber
\end{eqnarray}
\end{proof} As a quick corollary, we have $\{J_{ ij}, r\}=\{J_{ ij}, \pi^2\}=0$, and
\begin{eqnarray}
\{J_{ ij},J_{ i'j'}\}  & = &
\delta_{ i i'}J_{jj'}+\delta_{jj'}J_{ i i'}
-\delta_{ ij'}J_{j i'}-\delta_{j i'}J_{ ij'}.
\end{eqnarray} So $J_{ ij}$ s satisfy the commutation relations among the standard basis elements of the real Lie algebra $\frk{so}(2k+1)$. Then Lemma \ref{lemma2} may be paraphrased as follows: under the action of $J_{ ij}$ s, $x_ i$ s and $\pi_ i$ s transform as
$\frk{so}(2k+1)$ vectors, and $F_{ ij}$ s transform as $\frk{so}(2k+1)$ bi-vectors. It is then clear that  $J_{2k+2, 0}$, $J_{-1, 0}$, and $J_{2k+2, -1}$ all transform as $\frk{so}(2k+1)$ scalars; $J_{i, 2k+2}$, $J_{i, -1}$, and $J_{i, 0}$ all transform as $\frk{so}(2k+1)$ vectors. Then it is clear that the Poisson relations (\ref{cmtr}) hold whenever $J_{ ij}$ appears on its left hand side.

By using the identity $x_i F_{ij}=0$, one can check that $\{J_{2k+2, -1}, \vec
r\}=-\vec r$,  $\{J_{2k+2, -1}, r\}=-r$, $\{J_{2k+2, -1},
{1\over r}\}={1\over r}$, $\{J_{2k+2, -1}, \vec \pi\}=\vec \pi$, $\{J_{2k+2, -1}, r^2F_{ij}\}=0$.
That is, $J_{2k+2, -1}$ is the \emph{dimension operator} in physics.
It is then clear that the Poisson relations (\ref{cmtr}) hold whenever $J_{2k+2, -1}$ appears on its left hand side.

\vskip 10pt
The remaining verifications are divided into four cases.

\underline{Case 1}.
\begin{eqnarray}
\{J_{i, 0}, J_{j, 0} \}=-J_{ij}, &\{J_{i,0},  J_{2k+2, 0}\}=- J_{i, 2k+2},\cr
  \{J_{i,0}, J_{-1, 0}\}= - J_{i, -1},  &\{J_{2k+2,0}, J_{-1, 0}\}= - J_{2k+2, -1}.\nonumber
\end{eqnarray} But these are just the defining relations. So case 1 is done.

To check the remaining cases, it is convenient to introduce 
\begin{eqnarray}
X & := &r\pi^2+{Q\over r}, \quad Y:=r, \quad  W_i :=\{Y, J_{i,0}\}=x_i \cr
Z_i & := & \{X, J_{i,0}\} = x_i \pi^2 - 2\pi_i (\vec r\cdot  \vec \pi)+2r^2F_{ij}\pi_j-{Q\over r^2}x_i. \nonumber
\end{eqnarray}

\underline{Case 2}.
\begin{eqnarray}
\{J_{i,-1},J_{-1,0}\}= J_{i, 0}, & \{J_{i,0},J_{2k+2,0}\}= 0,\cr
\{J_{i,2k+2},J_{-1,0}\}= 0, & \{J_{i,2k+2},J_{2k+2,0}\}= -J_{i,0},\nonumber
\end{eqnarray}
or equivalently,
\begin{eqnarray}
\{W_i, Y\}=\{Z_i, X\}=0,\quad \{W_i, X\}=\{Z_i, Y\}=2J_{i,0}.
\end{eqnarray}
\begin{proof} It is clear that $\{W_i, Y\}=0$. Now $\{W_i,
X\}=r\{x_i, \pi^2\}=2J_{i,0}$. Next, using the identity
$F_{ij}x_j=0$, we have
\begin{eqnarray}
\{Z_i, Y\}&=&\{x_i\pi^2 - 2 \pi_i(\vec r\cdot \vec \pi)+2r^2F_{ij}\pi_j-{Q\over r^2}x_i , r\}\cr 
&=&x_i\{\pi^2, r\} - 2 \{\pi_i(\vec r\cdot \vec \pi), r\}-2rF_{i j}x_j \cr
&=&-2{x_i\over r}\vec r\cdot\vec \pi - 2\pi_i\{\vec r\cdot \vec \pi, r\} - 2 \{\pi_i, r\}\vec r\cdot \vec \pi \cr
&=&-2{x_i\over r}\vec r\cdot\vec \pi+2 \pi_i r + 2 {x_i\over r}\vec r\cdot \vec \pi \cr
&=&2J_{i, 0}.\nonumber
\end{eqnarray} Consequently, since $\{Z_i, r{1\over r}\}=0$, we get $\{{1\over r}, Z_i\} = {2\over r}\pi_i $. Finally
\begin{eqnarray}
\{r\pi^2, Z_i\} & = &  \{r, Z_i\}\pi^2+r\{\pi^2, Z_i\}\cr 
& = &  -2J_{i,0}\pi^2+r\{\pi^2, Z_i\}\cr
& = &-2J_{i,0}\pi^2+r\{\pi^2,  x_i\pi^2 - 2 \pi_i(\vec r\cdot \vec \pi)+2r^2F_{i j}\pi_j-{Q\over r^2}x_i \}\cr 
&=& -2J_{i,0}\pi^2+ r(\{\pi^2, x_i\}\pi^2 - 2 \{\pi^2, \pi_i(\vec r\cdot \vec \pi)\}\cr
&&+2\{\pi^2, r^2F_{i j}\pi_j\}-Q\{\pi^2, {x_i\over r^2}\} )\cr
&=& -2r\pi_i\pi^2+ r(-2\pi_i\pi^2 - 2 \{\pi^2, \pi_i\}(\vec r\cdot \vec \pi)+4\pi_i\pi^2\cr
&&+2\{\pi^2, r^2F_{i
j}\pi_j\}-Q\{\pi^2, {x_i\over r^2}\})\cr 
&=& r( - 2 \{\pi^2, \pi_i\}(\vec r\cdot \vec \pi)+2\{\pi^2, r^2F_{ij}\}\pi_j\cr
&&+2r^2F_{i j}\{\pi^2, \pi_j\}-Q\{\pi^2, {x_i\over r^2}\} )\cr
&=& r\left( - 2 \{\pi^2, \pi_i\}(\vec r\cdot \vec \pi)-2\{\pi^2, x_i\pi_j-x_j\pi_i\}\pi_j\right)\cr
&&+r\left(2r^2F_{i j}\{\pi^2, \pi_j\}-Q\{\pi^2, {x_i\over r^2}\} \right)\cr 
&=& r\left( 4 \pi_kF_{ki}\, \vec r\cdot \vec \pi-2\{\pi^2, x_i\pi_j\}\pi_j+2\{\pi^2,x_j\pi_i\}\pi_j\right)\cr
&&+r\left(4r^2F_{i j}F_{jk}\pi_k+2Q({\pi_i\over r^2}-{2x_i\over r^4}\vec r\cdot \vec \pi) \right)\cr
&=& r\left(4 \pi_kF_{ki}\, \vec r\cdot \vec \pi-2x_i\{\pi^2,
\pi_j\}\pi_j+2x_j\{\pi^2,\pi_i\}\pi_j\right)\cr
&&+r\left(4r^2F_{i j}F_{jk}\pi_k+2Q({\pi_i\over r^2}-{2x_i\over r^4}\vec r\cdot \vec \pi) \right)\cr
&=& r\left(-4r^2F_{ji}F_{jk}\pi_k+2Q({\pi_i\over r^2}-{2x_i\over r^4}\vec r\cdot \vec \pi) \right)\cr
&=& -{2Q\over r}\pi_i.\quad \mbox{by identity (\ref{Id4})}\nonumber
\end{eqnarray}
Therefore, $\{X, Z_i\}=\{r\pi^2+{Q\over r}, Z_i\} =0$.

\end{proof}

\underline{Case 3}.
\begin{eqnarray}
\{J_{i,-1},J_{j,0}\}= -\eta_{ij}J_{-1, 0},\quad
\{J_{i,2k+2},J_{j,0}\}= -\eta_{ij}J_{2k+2, 0}.
\end{eqnarray} or equivalently
\begin{eqnarray}
\{Z_i,J_{j,0}\}= -\eta_{ij}X,\quad
\{W_i,J_{j,0}\}= -\eta_{ij}Y.
\end{eqnarray}\begin{proof}
Half of the checking is easy:  $\{W_i, J_{j,0}\}=\{x_i, r\pi_j\}=r
\delta_{ij}=-\eta_{ij}Y$. To check the remaining half, we recall from step 2 that $\{Z_i, r\}=2J_{i,0}=2r\pi_i$, then
\begin{eqnarray}
\{Z_i, J_{j, 0}\} & = & r\{Z_i, \pi_j\}+2J_{i, 0}\pi_j\cr
& = & r\{x_i\pi^2 - 2 \pi_i(\vec r\cdot \vec \pi)+2r^2F_{i k}\pi_k-{Q\over r^2}x_i, \pi_j\}+2r\pi_i\pi_j\cr 
& = & r\left(\delta_{ij}\pi^2 -2x_i F_{k j}\pi_k\right)+r(2 F_{ij} \vec r\cdot \vec \pi-2\pi_i\pi_j)\cr
&&+r(-2r^2F_{ik}F_{k j}+\{2r^2F_{i k}, \pi_j\}\pi_k)+Q r\{\pi_j,{x_i\over r^2}\} +2r\pi_i\pi_j\cr
 & = & \delta_{ij}r\pi^2 -2rx_i F_{k j}\pi_k+2 rF_{ij} \vec r\cdot \vec \pi\cr
 &&+r(-2r^2F_{i k}F_{k j}+4x_j F_{ik}\pi_k-2r^2 \{\pi_j, F_{i k}\} \pi_k )+Q r\{\pi_j,{x_i\over r^2}\} \cr
& = & \delta_{ij}r\pi^2 -2r^3F_{ik}F_{kj}+2r(2x_j F_{i k}+x_iF_{j k}+x_k F_{ij}+r^2 \nabla_j F_{i k})\pi_k\cr 
& &+Q r\{\pi_j,{x_i\over r^2}\} \cr 
& = & \delta_{ij}r\pi^2 +2r^3F_{ k i}F_{k j}-Q r\left({\delta_{ij}\over r^2}-2{x_i x_j\over r^4}\right)\quad \mbox{by identity (\ref{Id2})}\cr 
&=& \delta_{ij}(r\pi^2+{Q\over r^2})\quad \mbox{by identity (\ref{Id4})}\cr
&= &-\eta_{ij}X.\nonumber
\end{eqnarray}
\end{proof}

\underline{Case 4}.
\begin{eqnarray}
\{J_{i, -1},J_{j, -1}\} = -J_{ij},\quad
\{J_{i, 2k+2},J_{j, -1}\} = -\eta_{ij}J_{2k+2, 0},\quad
\{J_{i, 2k+2},J_{j, 2k+2}\} = J_{ij}\nonumber
\end{eqnarray}
or equivalently,
\begin{eqnarray}
\{Z_{i},Z_{j}\} =\{W_i, W_j\}= 0,\hskip 10pt
\{Z_{i},W_{j}\}
=-2\left(\eta_{ij}J_{2k+2, -1}+J_{ij}\right).
\end{eqnarray}
\begin{proof} It is clear that $\{W_i, W_j\}= 0$ because
$W_i=x_i$. Next,
\begin{eqnarray}
\{Z_{i},W_{j}\} & = & \{x_i \pi^2 - 2\pi_i \,\vec r\cdot \vec \pi+2r^2F_{i  k}\pi_k-{Q\over
r^2}x_i , x_j\}\cr 
&=&x_i \{\pi^2, x_j\} - 2\{\pi_i\, \vec r\cdot \vec \pi, x_j\}+2r^2F_{i k}\{\pi_k, x_j\}\cr 
&=&-2x_i \pi_j - 2\pi_i \{\vec r\cdot \vec \pi, x_j\}- 2\{\pi_i, x_j\} \vec r\cdot \vec \pi-2r^2F_{i j}\cr 
&=&-2x_i \pi_j +2\pi_i x_j+2\delta_{ij} \vec r\cdot \vec \pi-2r^2F_{ij} \cr 
&=&-2\left(x_i \pi_j-x_j\pi_i+r^2 F_{ij}\right)+2\delta_{ij}\, \vec r\cdot \vec \pi\cr 
& = & -2\left(\eta_{ij}J_{2k+2, -1}+J_{ij}\right).\nonumber
\end{eqnarray}Finally, using results from case 2 and case 3, we have
\begin{eqnarray}
-\{Z_i, Z_j\}& = & \{\{J_{i,0}, X\}, Z_j\}\cr
&= &\{\{J_{i,0}, Z_j\}, X\}+\{J_{i,0}, \{X, Z_j\}\}
\cr &=&\{\eta_{ij}X, X\}+\{J_{i,0}, 0\}= 0.\nonumber
\end{eqnarray}
\end{proof}
\begin{rmk}
The $J_{AB}$s introduced here define a map $J$:  $\mathcal M^\sharp \to \frk{so}^*(2, 2k+2)$. The image of $J$, after adding some missing points, is the union of some co-adjoint orbits of  $\frk{so}(2, 2k+2)$, one for each $\mu\in \mathbb R$, and the orbit with $\mu=0$ lies in the nilpotent cone. 
\end{rmk}
\section{An extension of the classical MICZ Kepler problems}
In view of the work done in Ref. \cite{meng11a}, the Poisson realization of the conformal algebra $\frk{so}(2, 2k+2)$ on a magnetic leave naturally yields a classical generalized Kepler problem associated with the Jordan algebra $\Gamma(2k+1):=\mathbb R\oplus \mathbb R^{2k}$. The goal here is to describe this new classical generalized Kepler problems, as envisaged in Ref. \cite{meng07}. 

Recall from Ref. \cite{meng11a} that the {\bf classical universal Hamiltonian}, {\bf classical universal angular momentum},  and the {\bf classical universal Lenz vector} are
\begin{eqnarray}\label{universalHcla}
\mathcal H={{1\over 2} \mathcal X_e-1\over \mathcal Y_e},\quad {\mathcal L}_{u, v}=\{{\mathcal L}_u, {\mathcal L}_v\}, \quad \mathcal A_u:={1\over 2}\left(\mathcal X_u-\mathcal Y_u{\mathcal X_e\over \mathcal Y_e}\right)+{\mathcal Y_u\over \mathcal Y_e}.\end{eqnarray} respectively \footnote{The hamiltonian $\mathcal H$ here is collective in the sense of Guillemin-Sternberg, i.e., it is a function of the components of the moment map $J$ introduced in section 5.}. Here, $u, e$ are the elements of the Jordan algebra $\Gamma(2k+1)$ with $e$ being the identity element. In our Poisson realization of $\frk{so}(2, 2k+2)$ on magnetic leaves, 
\begin{eqnarray}\left\{
\begin{array}{l}
{\mathcal L}_{e_i}=J_{i,0}=r\pi_i, \cr
{\mathcal Y}_e = Y=r,\cr
{\mathcal X}_e = X =r\pi^2+{Q\over r},\cr
{\mathcal X}_{e_i} = -Z_i=-x_i \pi^2 +2\pi_i (\vec r\cdot  \vec \pi)-2r^2F_{ij}\pi_j+{Q\over r^2}x_i,\cr
{\mathcal Y}_{e_i} = W_i = x_i.
\end{array}\right.
\end{eqnarray} Here $e_1$, \ldots, $e_{2k+1}$ are the standard basis vectors for $\mathbb R^{2k+1}$.
Therefore, we have a generalized classical Kepler problem for which, the hamiltonian is
\begin{eqnarray}\label{hamiltonian}
\fbox{$H={1\over 2}\pi^2+{Q\over 2r^2}-{1\over r}$}, 
\end{eqnarray} the angular momentum $L={1\over 2}\sum_{i,j}L_{ij}e_i\wedge e_j$ with $L_{ij}=L_{e_j, e_i}$ is
\begin{eqnarray}
\fbox{$L=\vec r\wedge \vec \pi+r^2F$}, \quad \mbox{where } F:={1\over 2}\sum_{ij} F_{ij }\, e_i\wedge e_j
\end{eqnarray} and the Lenz vector $\vec A = \sum_i A_ie_i$ with $A_i=A_{e_i}$ is 
\begin{eqnarray}
\fbox{$\vec A= \vec \pi \lrcorner L+{\vec r\over r}$}. 
\end{eqnarray} 
Notice that $$Q={1\over 2k}\sum_{1\le a, b \le 2k}(\gamma_{ab})^2$$ is (up to the scale $1\over k$) the Casimir operator on $\frk{g}:= \frk{so}(2k)$, but is viewed here as a function on $E_\sharp=P\times _G \frk{g}^*$ in the following sense:
for any point $[(p, \xi)]\in E_\sharp$, we have 
$$
Q([(p, \xi)])=\langle \xi, Q\rangle.
$$
By using the basic Poisson bracket relations (\ref{PoissonRel}), with the understanding of $F_{ij}$ in the sense of Eq. (\ref{F-mean}), the Hamilton's equation $f'=\{f, H\}$ for the basic functions $x_i$, $\pi_i$, and $T_\alpha$ becomes
\begin{eqnarray}
\left\{
\begin{array}{l}
x_i' =\pi_i,\cr
\pi_i' = -{x_i\over r^3}+Q {x_i\over r^4}+ \pi_j F_{ji},\cr
T_\alpha' =\pi_k [iA_k,  T_\alpha]=-C_{\alpha\beta}^\gamma \pi_k A_k^\beta T_\gamma.
\end{array}
\right.
\end{eqnarray} 
Therefore, the equation of motion becomes
\begin{eqnarray}\label{EqnM}
\left\{
\begin{array}{l}
\vec r'' = -{\vec r\over r^3}+Q {\vec r\over r^4}+ \vec r'\lrcorner F,\cr
\\
T_\alpha' =  [i\vec r'\lrcorner A,  T_\alpha].
\end{array}
\right.
\end{eqnarray} 
Here, $\lrcorner$ is the interior product. A coordinate-free and gauge-free formulation for the equation of motion, which should exist in the first place, can be obtained from the above equation of motion. To describe it, we let $\vec r$: $\mathbb R\to X:=\mathbb R^{2k+1}_*$ be a smooth map, and $\xi$ be a smooth lifting of $\vec r$:
\begin{eqnarray}
\begin{array}{ccc}
 & & E_\sharp\\
 \\
 & \xi \nearrow &\Big\downarrow\\
 \\
 \mathbb R &\buildrel\vec r\over\longrightarrow & X
\end{array}
\end{eqnarray}
Then the 2nd equation of (\ref{EqnM}) says that $\xi$ is a covariantly constant section of the pullback bundle over $\mathbb R$, i.e.,  
 $$
 {D\over dt}\xi =0.
 $$
Let $Ad_P$ be the adjoint bundle $P\times_G\frk{g}\to {\mathbb R}_*^{2k+1}$, $d_\nabla$ be the canonical connection, i.e., the generalized Dirac monopole on $\mathbb R^{2k+1}_*$. Then the curvature ${\Omega}:=d_\nabla^2$ is a smooth section of the vector bundle $\wedge^2T^*{\mathbb R}_*^{2k+1}\otimes Ad_P$. (With the trivialization of $P\to {\mathbb R}_*^{2k+1}$ chosen in section \ref{monopoles}, locally $\Omega$ can be represented by ${1\over 2}\sqrt{-1}F_{jk}\, dx^j\wedge dx^k$.) 
 Since $Q$ becomes $|\xi|^2\over k$, $Q$ is a constant, i.e., independent of the time. Then the equation of motion (\ref{EqnM}) can be reformulated as equation
\begin{eqnarray}\label{EqnMF}
\fbox{$\left\{
\begin{array}{l}
\vec r'' = -{\vec r\over r^3}+{|\xi|^2\over k} {\vec r\over r^4}+ <\xi, \vec r'\lrcorner \Omega>,\cr
\\
{D\over dt}\xi =0.
\end{array}
\right.$}
\end{eqnarray} Here $<, >$ refers to the pairing of the adjoint bundle and its coadjoint bundle, and 2-forms are identified with 2-vectors via the standard euclidean structure of $\mathbb R^{2k+1}$. 

Eq. (\ref{EqnMF}) can be viewed as a dynamical equation on $E_\sharp$, but then it is not super integrable. By restricting to the magnetic leave with magnetic charge $\mu$, Eq. (\ref{EqnMF}) defines a super integrable model \footnote{Here the word ``super integrable" means that the number of functionally independent conserved quantities is equal to the the dimension of the phase space minus one.} which generalize the classical MICZ Kepler problem. This super integrable model shall be referred to as the {\bf classical Kepler problem with magnetic charge $\mu$ in dimension $2k+1$}. In dimension $5$, it is essentially Iwai's $\mr{SU}(2)$-Kepler problem, cf. Ref. \cite{Iwai90}.

It is not hard to see that $L^2-kQ=|\vec r\wedge \vec r'|^2$, then  $L^2-kQ>0$ for the non-colliding orbits. Note that $kQ=\mu^2$ on $\mathcal O_\mu$. By using the quadratic relations (\ref{QuadraticRel}) one can check that, for the non-colliding orbits of the classical Kepler problem with magnetic charge $\mu$, the total energy 
\begin{eqnarray}
\fbox{$H=-{1-A^2\over 2(L^2-\mu^2)}$}. 
\end{eqnarray} Here, $A^2$ is the length square of the Lenz vector $\vec A$, i.e., $A^2=\sum_i A_i^2$, and $L^2$ is the length square of $L$, i.e., $L^2=\sum_{i<j}L_{ij}^2$. It is expected from Ref. \cite{meng11} (though not from Ref. \cite{Iwai90}) that even in this generalized model, a non-colliding orbit remains an ellipse or a parabola or a branch of hyperbola when the total energy $H$ is negative or zero or positive respectively. The details will be presented elsewhere. 

\vskip 5pt
An interesting direction to explore is to work out the geometric quantization of the models introduced here so that one can reproduce the quantum models introduced in Ref. \cite{meng07}. We expect that the earlier work carried out by I. Mladenov and V. Tsanov \cite{Mladenov} for the Kepler problems in higher dimensions or the MICZ Kepler problems shall serve a good guidance in this exploration.

\end{document}